\documentclass[12pt]{amsart}
\usepackage[utf8]{inputenc}
\usepackage{amsmath}
\usepackage{amssymb}
\usepackage[alphabetic]{amsrefs}
\usepackage{geometry,verbatim}
\usepackage{graphicx,color,hyperref}
 \usepackage{bbm,fullpage}
 
 \newcommand{\figref}{Figure \ref}
 \newtheorem{lemma}{Lemma}[section]
\newtheorem{theorem}[lemma]{Theorem}
\newtheorem{thm}[lemma]{Theorem}

\newtheorem{prop}[lemma]{Proposition}
\newtheorem{cor}[lemma]{Corollary}
\newtheorem{proposition}[lemma]{Proposition}

\newtheorem*{thm*}{Theorem}
\newtheorem*{def*}{Definition}
\theoremstyle{definition}

\newtheorem{remark}[lemma]{Remark}
\newtheorem{rem}[lemma]{Remark}

\newtheorem{example}[lemma]{Example}

\newcommand\R{\mathbb{R}}
\newcommand\RR{\mathbb{R}}

\newcommand\sym{\operatorname{sym}}

\begin{document}  
    
\title{Quantum Entanglement, Symmetric Nonnegative Quadratic Polynomials and Moment Problems}

\author[G.~Blekherman]{Grigoriy Blekherman}
\address{Grigoriy Blekherman: School of Mathematics, Georgia Tech, 686 Cherry
  Street, Atlanta, Georgia, 30332, United States of America;
  {\normalfont \texttt{greg@math.gatech.edu}}}

\author[B.H.~Madhusudhana]{Bharath Hebbe Madhusudhana}
\address{Bharath Hebbe Madhusudhana: Fakult\"{a}t f\"{u}r Physik, Ludwig-Maximilians-Universit\"{a}t M\"{u}nchen,
Schellingstr. 4,
80799 M\"{u}nchen, Germany;
 {\normalfont \texttt{Bharath.Hebbe@physik.uni-muenchen.de}}}
\date{\today}

 \begin{abstract}
Quantum states are represented by  positive semidefinite Hermitian operators with unit trace, known as density matrices. An important subset of quantum states is that of separable states, the complement of which is the subset of \textit{entangled} states. 
We show that the problem of deciding whether a quantum state is entangled can be seen as a moment problem in real analysis. Only a small number of such moments are accessible experimentally, and so in practice the question of quantum entanglement of a many-body system (e.g, a system consisting of several atoms) can be reduced to a truncated moment problem.
By considering quantum entanglement of $n$ identical atoms we arrive at the truncated moment problem defined for symmetric measures over a product of $n$ copies of unit balls in $\RR^d$. We work with moments up to degree $2$ only, since these are most readily available experimentally. We derive necessary and sufficient conditions for belonging to the moment cone, which can be expressed via a linear matrix inequality of size at most $2d+2$, which is independent of $n$. The linear matrix inequalities can be converted into a set of explicit semialgebraic inequalities giving necessary and sufficient conditions for membership in the moment cone, and show that the two conditions approach each other in the limit of large $n$. The inequalities are derived via considering the dual cone of nonnegative polynomials, and its sum-of-squares relaxation. We show that  the sum-of-squares relaxation of the dual cone is asymptotically exact, and using symmetry reduction techniques \cites{MR2067190,BR}, it can be written as a small linear matrix inequality of size at most $2d+2$, which is independent of $n$. For the cone of symmetric nonnegative polynomials with the relevant support we also prove an analogue of the half-degree principle for globally nonnegative symmetric polynomials  \cites{MR1996126,MR2864859}.  \end{abstract}

\maketitle
\section{Introduction}

We consider a problem that lies in the intersection of real analysis, real algebraic geometry and quantum entanglement.  We begin with a brief introduction to quantum entanglement, explaining how it is related to real analysis and algebraic geometry. We then define and motivate the problem considered in this paper. 

In quantum mechanics, the state of a physical system is represented by a \textit{density matrix}, which is, by definition, a positive semi-definite, Hermitian operator with trace one. As a convention, we drop the trace condition so that the space of density matrices is a convex cone. Since any positive semidefinite matrix can be rescaled to have trace one, this does not affect the results of the paper. Let $T_n$ be the convex cone of all positive semi-definite Hermitian operators acting on $\mathbb{C}^n$. By our convention, an element of $T_n$ can be rescaled to an $n \times n$ density matrix. Physical systems that correspond to $n=2$ are known as \textit{qubits}. 

A natural way to construct larger density matrices is by using tensor products.  In physics this corresponds to combining several subsystems into a large system. As an example, we will consider a system with two subsystems represented by density operators in $T_m$ and $T_n$ respectively.  The tensor product gives a map $T_m \times T_n \rightarrow T_{mn} $, where $\rho_1\in T_m$ and $\rho_2\in T_n$ are mapped to $\rho_1\otimes \rho_2 \in T_{mn}$, for $m, n \in \mathbb{N}$.  A convex subset $\Omega_{nm}$ of $T_{nm}$ is generated by the conical hull of the image of this map:
\begin{equation*}
\Omega_{mn} = \operatorname{conical.hull}\{\rho_1 \otimes \rho_2: \rho_1 \in T_m \ \ \rho_2 \in T_n \}.
\end{equation*}
An element $\rho\in \Omega_{mn}$ can be obtained from a measure $\mu$ over $T_{m}\times T_n$ via $\rho = \int_{T_m \times T_n}  \rho_1\otimes\rho_2 \, d\mu $. The density matrices in $T_{m}$ are states of an $m$-dimensional system; the density matrices in $T_{n}$ are states of an $n$-dimensional system, and the density matrices in $T_{mn}$ are states of the composite of the two systems. The subset $\Omega_{mn}\subset T_{mn}$ consists of \textit{separable} states. That is, a density matrix in $\Omega_{mn}$ represents a separable state and a density matrix in $T_{mn}\setminus \Omega_{mn}$ represents an \textit{inseparable state}, also known as an \textit{entangled state} of the two subsystems. While the notion of quantum \textit{non-locality} \cites{ PhysRevA.40.4277, Peres1997} is different from inseparability, the latter is used synonymously with entanglement (see for instance, the  recent review \cites{RevModPhys.81.865}). Quantum entanglement is not only fundamental to quantum physics, it also has applications in quantum communications, quantum metrology and quantum computation \cite{RevModPhys.81.865}.  

One of the basic problems in quantum entanglement is the characterization of the cone $\Omega_{mn}$. In the special case of rank-1 density matrices, also known as pure states, this problem has a complete solution \cite{GISIN1991201}. The $mn$ elements of the vector in the range of a rank-1 matrix $\rho \in T_{mn}$ can be written as a $m\times n$ matrix in a natural way. $\rho \in \Omega_{mn}$ iff this matrix also has rank-1. The general problem of deciding whether an element $\rho\in T_{mn}$ is in $\Omega_{mn}$ has been shown to be NP-hard in the complexity parameter $mn$ \cites{2003quantph3055G}. However, several special instances of this problem can be solved exactly. When $m=2, n=2$ and when $m=2, n=3$, the cone $\Omega_{mn}$ has been completely characterized \cites{PhysRevLett.77.1413, HORODECKI19961}, using positive partial transpose criterion. The general case of larger dimensions, or greater number of subsystems, is open, although there are a number of partial results \cites{Recent_Development_1, Recent_Development_2, Recent_Development_3, Recent_Development_4, Recent_Development_5, PhysRevA.89.062110}.  

Our approach is to convert the study of quantum entanglement into a truncated moment problem. The relation between quantum entanglement and the truncated moment problem has been recognised rather recently \cite{PhysRevA.96.032312, doi:10.1137/15M1018514, SDP_entanglement_1} and therefore, the full potential of the state of the art in the truncated moment problem has not been utilized to address the entanglement problem.  The elements of $\rho \in \Omega_{mn}$ are quadratic moments of a measure over $T_{m}\times T_n$, where each quadratic monomial has one variable from $T_m$ and one variable from $T_n$. Therefore, characterizing the cone $\Omega_{mn}$ can be viewed as a truncated moment problem. In fact, this equivalence between deciding whether a quantum state is entangled and a truncated moment problem generalizes to quantum systems with arbitrary number of subsystems. Although quantum entanglement has been studied extensively over the past few decades, this formulation, which allows us to bring in the tools of real algebraic geometry, has not been sufficiently explored. In this work, we consider a class of truncated moment problems that arise from the study of quantum entanglement and show that the resulting criteria on entanglement are stronger than the existing ones.  

The most commonly studied physical system is the so-called ``two-level system" or a ``qubit", described by a density matrix in $T_2$. Theoretically, this is the simplest non-trivial example of quantum mechanical systems, and experimentally, there have been several platforms to realize such systems including trapped atoms, photons, and circuit QED systems. Therefore, we consider a truncated moment problem that arises from the study of entanglement in a system consisting of $n$ subsystems, each described by a density matrix in $T_2$. The space of all density matrices of the composite system is $T_{2^n}$. The subspace of separable states is 
\begin{equation*}
\Omega_{2^n} = \text{conical.hull}\{\rho_1 \otimes \rho_2\otimes \cdots \otimes \rho_n: \rho_i \in T_2 \}.
\end{equation*}
From the above discussion we see that $\Omega_{2^n}$ is a moment cone of measures defined over $T_2\times \cdots \times T_2$, where we take degree $n$ moments, and each moment monomial has one variable from the coordinates of each $T_2$. We observe that $T_2$ has a simple description: any self-adjoint $2\times 2$ complex matrix can be written as 
\begin{equation*}
\rho = \left(\begin{array}{cc}
w+z & x-iy \\
x+iy & w-z
\end{array}\right),
\end{equation*}
for $x, y, z, w \in \mathbb{R}$ and $w\geq 0$. The positive semidefiniteness condition is given by $w^2 -x^2-y^2-z^2 \geq 0$. Due to multilinearity of tensors, 
we may restrict ourselves to measures defined on the compact section $\tilde{T}_2$ of $T_2$ with the hyperplane $w=1$.  We see from above that $\tilde{T}_2$ is $D^3$, the closed unit ball in $\mathbb{R}^3$ also known as the Bloch sphere. Since we dehomogenized by setting $w=1$ we now consider moments up to degree $n$ on $\tilde{T}^n_2$, where each moment monomial has \textit{at most} $1$ variable from the coordinates of each $\tilde{T}_2$.
%In other words, a \gb{$w=1$?} slice of $T_2$ is 
Therefore, $\Omega_{2^n}$ is a moment cone of such moments for measures defined over the space $K_n^3 =  D^3 \times \cdots \times D^3$.

There are two characteristics of a real physical situation that simplify this problem. First,  in a laboratory, not every element of a density matrix $\rho \in T_{2^n}$ can be retrieved. Therefore, the decision of whether $\rho$ can be in $\Omega_{2^n}$ has to be made based on a subset of the elements of $\rho$. In other words, the problem requires characterization of a projection of the cone of moments over $K_n^3$, usually consisting of very few moments. We do not consider the characterization of the cone $\Omega_{2^n}$ here; instead, we consider the characterization of its projection. That is, corresponding to a linear space $L$, we consider the decision problem of deciding whether a point in $\Pi_L(T_{2^n})$ is in $\Pi_L(\Omega_{2^n})$. The linear space $L$ is defined by taking all the relevant quadratic moments, since these are most readily measured in the lab \cite{vladan_2015,Hoang9475}. And secondly, the atoms used in such experiments are called \textit{bosons}, meaning, they possess an exchange symmetry. That is, the density matrix $\rho$ is invariant under the action of the symmetric group $S_n$, which permutes the atoms, and so is any measure $\mu$ that generates the density matrices in $\Omega_{2^n}$. Therefore, the problem is to characterize the moment cone of \textit{symmetric} measures defined over $K_n^3$. 

We now consider a generalized version of the above problem in physics. In particular we use $D^d$ instead of $D^3$, although the structure is not physically relevant for any value of $d$ other than $3$.

\subsection{Problem Statement and Results}

Let $D^d\subset  \R^d$ be the closed unit ball in a $d$ dimensional real vector space. That is, $D^d=\{\mathbf{v}: \ \ \mathbf{v}\in  \R^d, \ ||\mathbf{v}||\leq 1\}$. We define $K^d_n$ as the product of $n$ such unit balls:
\begin{equation*}
K^d_n = D^d \times \cdots \times D^d \subset  \R^{nd}.
\end{equation*}
Points in $K_n^d$ can be represented by a $n$-tuple of $d$-dimensional vectors $(\mathbf{v}_1, \mathbf{v}_2, \cdots, \mathbf{v}_n )$, where, $\mathbf{v}_i = (v_{i,1}, v_{i,2} \cdots v_{i,d} )\in D^d$. The symmetric group $S_n$ (i.e., the group of all permutations of the set $\{1, 2, \cdots, n\}$) acts on $K_n^d$ by permuting the vectors $\mathbf{v}_i$. For $\sigma \in S_n$, its action on $K_n^d$ is given by $\sigma\circ (\mathbf{v}_1, \mathbf{v}_2, \cdots, \mathbf{v}_n )=(\mathbf{v}_{\sigma(1)}, \mathbf{v}_{\sigma(2)}, \cdots,  \mathbf{v}_{\sigma(n)} ) $. We refer to a measure $\mu$ defined on $K_n^d$ as a \textit{symmetric measure} if it is invariant under the action of $S_n$. That is, if $A\subset K_n^d$, 
\begin{equation*}
\mu(A)= \mu(\sigma\circ A), \,\, \forall\ \sigma \in S_n.
\end{equation*} 
In this work, we consider the truncated $K$-moment problem for symmetric measures over $K_n^d$. Moments of symmetric measures are also invariant under coordinate permutations.

Let $V_{n,d}$ the vector space of real square-free polynomials in $d\cdot n$ variables $\mathbf{x_1},\cdots, \mathbf{x_n}$, $\mathbf{x_i}=(\xi_{i,1},\cdots, \xi_{i,d})$ which is spanned by the following symmetric polynomials: $$1, \,\,\,\, s_{\alpha}=\sum_{i=1}^n \xi_{i,\alpha}, \,\,\, 1\leq \alpha\leq d \hspace{.5cm}\text{and} \hspace{.5cm} s_{\alpha\alpha}=\sum_{i \neq j} \xi_{i,\alpha}\xi_{j,\alpha}, \,\, 1\leq \alpha\leq d.$$

A polynomial $Q$ in $V_{n,d}$ can be uniquely written as 
$$Q(\mathbf{x})=A_0+\sum_{\alpha=1}^d A_\alpha s_\alpha+\sum_{\alpha=1}^d A_{\alpha\alpha}s_{\alpha\alpha}.$$
The dimension of $V_{n,d}$ is therefore $2d+1$. This choice of basis polynomials $s_\alpha$ and $s_{\alpha\alpha}$ is dictated by the underlying physics. The moments are symmetric because the particles are bosons, possessing exchange symmetry, they are square-free since tensor product of different vector spaces only include multilinear (or square-free) monomials, and we only look at quadratic moments because the corresponding moments can be measured in a laboratory \cite{vladan_2015, Hoang9475}.

We sometimes refer to polynomials in $V_{n,d}$ by the $(2d+1)$-tuple of coefficients $(A_0,A_{\alpha},A_{\alpha\alpha})$.

We use $m_{0}, m_\alpha$ and $m_{\alpha\alpha}$ to denote the corresponding moments:
\begin{equation*}
m_0 = \int_{K_n^d}1 \, d\mu, \,\,\,\, m_{\alpha}= \int_{K^d_n}s_{\alpha}\,d\mu, \,\,\,\,m_{\alpha\alpha}= \int_{K^d_n}s_{\alpha\alpha}\, d\mu.
\end{equation*}
The moment sequence $\beta_{\mu}=(m_0,m_{\alpha},m_{\alpha\alpha})$ lies in $\R^{2d+1}$. We define $C_{n,d}\subset \R^{2d+1}$ to be the set of all moment sequence coming from measures on $K_n^d$. Since the product of the unit balls $K_n^d$ is compact, it follows that $C_{n,d}$ is a closed convex cone (cf. \cite[Exercise 4.17]{BPT}).

The dual cone of $C_{n,d}$ is the cone $P_{n,d}$ of symmetric polynomials in $V_{n,d}$ nonnegative on $K_n^d$: \cite[Chapter 4]{BPT}.
\begin{equation*}
P_{n,d}= \left\{Q \in V_{n,d}: \ \  Q(\mathbf{x}) \geq 0 \ \text{for all}\  \mathbf{x} \in K_n^d  \right\}.
\end{equation*}
We consider the problem of characterizing the cones $C_{n,d}$ and $P_{n,d}$. Our main result are four explicit Linear Matrix Inequalities (LMI) of size $2d+2$, which give necessary and sufficient conditions for belonging to $C_{n,d}$ and $P_{n,d}$ (Theorems \ref{thm:sos} and  \ref{thm:approx} for $P_{n,d}$, and Theorem \ref{thm:maindual} for $C_{n,d}$). The linear matrix inequalities for $C_{n,d}$ and $P_{n,d}$ are dual to each other.  These necessary and sufficient criteria can be seen to approach the same limit as the number $n$ of unit balls approaches infinity (see Remark \ref{rem:limit}). For the cone $C_{n,d}$ we were able to convert the linear matrix inequality into an explicit set of $2^d+1$ semialgebraic inequalities describing necessary and sufficient conditions for membership in $C_{n,d}$. %These criteria can be seen to approach the same limit as the number $n$ of unit balls approaches infinity (see Remark \ref{rem:limit}).

\begin{thm*}[Theorems \ref{thm:maindual} and \ref{thm:maindual2}]
A non zero vector $(z_0, z_1,\cdots, z_d, z_{11}, z_{22},\cdots, z_{dd})$ lies in $C_{n,d}$ if $z_0>0$ and
\begin{equation*}
    \sum_{\alpha=1}^d\max \left\{\frac{z_0z_{\alpha\alpha}}{n}, \frac{z_{\alpha}^2}{n}\right\} \leq  z_0^2\left(\frac{n-1}{n}\right)^2+ \sum_{\alpha=1}^d \frac{(n-1)z_0z_{\alpha\alpha}}{n^2}
\end{equation*}
Moreover, if the inequalities below are violated:
\begin{equation*}
  z_0>0  \hspace{.3cm} \text{and} \hspace{.3cm}  \sum_{\alpha=1}^d\max \left\{\frac{z_0z_{\alpha\alpha}}{n-1}, \frac{z_{\alpha}^2}{n}\right\} \leq  z_0^2+ \sum_{\alpha=1}^d \frac{z_0z_{\alpha\alpha}}{n}
\end{equation*}
Then the non zero vector $(z_0, z_1,\cdots, z_d, z_{11}, z_{22},\cdots, z_{dd}) \notin C_{n,d}$. Each of the above systems of inequalities may also be expressed by a linear matrix inequality of size $2d+2$.
\end{thm*}

The necessary conditions in the physically relevant case of $d=3$ appeared in \cite{PhysRevLett.99.250405} in the language of spin moments (see Section \ref{sec:last} for more details). For other related results from physics, see \cite{PhysRevLett.107.180502}, \cite{1367-2630-19-1-013027} and \cite{PhysRevLett.86.4431}. As we explain below the necessary conditions are obtained from the sum-of-squares approximation to the cone $P_{n,d}$. The sufficient conditions are new, and they allow us to estimate the tightness of approximation for large $n$, which was previously unknown.

We now give a more detailed description of our results. First we show that nonnegativity on $K^d_{n}$ of a polynomial $Q \in V_{n,d}$, can be established by testing its values on much smaller subset of $K^d_{n}$. It is easy to see that $Q\in V_{n,d}$ is nonnegative on $K_n^d$ if and only if $Q$ is nonnegative on the product of unit spheres $(S^{d-1})^n$. Furthermore we show that $Q(\mathbf{x})\geq 0$ for all $\mathbf{x}=(\mathbf{x_1},\cdots,\mathbf{x_n})\in (S^{d-1})^n$ if and only if $Q(\mathbf{x})\geq 0$ for all $\mathbf{x}$ with at most $2d$ of $\mathbf{x_i}$ distinct.  
\begin{thm*}[Theorem \ref{thm:halfdeg}] A polynomial $Q \in V_{n,d}$ is nonnegative on $S^{d-1}\times \cdots \times S^{d-1}$ if and only if $Q(\mathbf{x}_1, \cdots, \mathbf{x}_n)$ is nonnegative for all sets of $n$ points $\mathbf{x}_1, \cdots, \mathbf{x}_n$ on $S^{d-1}$ with at most $2d$ of them distinct.
\end{thm*}

The above theorem is in the spirit of the half-degree principle for globally nonnegative symmetric polynomials \cites{MR1996126,MR2864859}. Our problem is different in two ways. First, we are concerned here nonnegativity of polynomials over a compact subset, $K_n^d$, and second, the symmetry in our problem permutes only $d$-tuples of coordinates. This is a smaller group of transformations, compared to the usual symmetric case. Symmetries of such type have also been considered in \cite{MR3424059}, where such polynomials were called mutlisymmetric. However our bounds are sharper than those of \cite{MR3424059}. %Furthermore, we show in Theorem \ref{tight_thm} that the bound of $2d$ cannot be improved, when $d=2$.

Secondly, we provide an asymptotically tight characterization of the cones $C_{n,d}$ and $P_{n,d}$ using sums of squares approximations. For $i=1,\cdots,n$ define $p_i=1-\sum_{j=1}^d \xi^2_{j,i}.$  If $Q\in V_{n,d}$ can be written as 
$$Q(\mathbf{x})=\sum_{i=1}^r \ell_i^2(\mathbf{x})+\sum_{i=1}^d \lambda_i(1-p_i(\mathbf{x})),$$
where $\ell_i$ are linear polynomials and $\lambda_i \geq 0$, then $Q$ is clearly nonnegative on $K_n^d$.
Define $\Sigma_{n,d}$ to be the cone of such polynomials, whose nonnegativity can be certified via sums of squares. One can view $\Sigma_{n,d}$ as the first level of the Lasserre (or Sum-of Squares) hierarchy for $K^d_{n}$ \cite{lasserre2001global,MR2500468,MR3395550,parrilo2000structured}. See Section \ref{sec:1.2} for more details. We have the inclusion:
$$\Sigma_{n,d} \subseteq P_{n,d}.$$

%\gb{Discuss equality here..}
As we show in Section \ref{sec:asymp} the containment is already strict when $d=1$. Using this, strict containment for larger $d$ also follows by considering polynomials with $A_{\alpha}$ and $A_{\alpha\alpha}$ equal to $0$ for $\alpha\geq 2$.
For the analogous results on equality between globally nonnegative symmetric polynomials and sums of squares see \cite{MR3464065}.

Using symmetry reduction techniques \cites{MR2067190,BR} we show that $\Sigma_{n,d}$ can be characterized by a \textit{linear matrix inequality} (LMI), of size $2d+2$, which is independent of $n$. Now define $\Sigma'_{n,d}$ as a slightly expanded version of $\Sigma_{n,d}$:
$$ \Sigma'_{n,d} = \{(A_0, A_{\alpha}, A_{\alpha\alpha}): \ A_0+\frac{n-1}{n}\sum_{\alpha=1}^dA_{\alpha}s_\alpha+  \frac{n-1}{n}\sum_{\alpha=1}^dA_{\alpha\alpha}s_{\alpha\alpha} \in \Sigma_{n,d}\}.$$

Our main result on sum of squares approximations is that the reverse inclusion holds for $\Sigma'_{n,d}$ and $P_{n,d}$ :

\begin{thm*}[Theorem \ref{thm:approx}]%\label{thm:asympt}
We have the following inclusions:
$$\Sigma_{n,d}\subseteq P_{n,d}\subseteq \Sigma'_{n,d}.$$
Moreover, both cones $\Sigma_{n,d}$ and $\Sigma'_{n,d}$ can be described by a linear matrix inequality of size at most $2d+2$.
\end{thm*}

As the number of unit balls $n$ grows the cones $\Sigma_{n,d}$ and $\Sigma'_{n,d}$ approach each other at the rate $\frac{1}{n}$ and therefore, they give an asymptotically exact characterization of $P_{n,d}$. Therefore sums of squares give an asymptotically tight approximation of nonnegative polynomials. A similar result for fully symmetric globally nonnegative polynomials of degree $4$ was established in \cite{BR}, and for fully symmetric even degree $6$ polynomials in \cite{MR900345}. Using semidefinite duality and Schur complement we derive the LMI and semialgebraic descriptions of the cone $C_{n,d}$ in Theorems \ref{thm:maindual} and \ref{thm:maindual2}.

\subsection{Relation to previous work on Lasserre/Sum of squares hierarchies.}\label{sec:1.2} The cone $\Sigma_{n,d}$ is simply the first level of the Lasserre or sum of squares hierarchy for $K^d_{n}$ \cite{lasserre2001global,MR2500468,MR3395550,parrilo2000structured}. We observe that the size of the linear matrix inequalities produced by the standard hierarchy grows with the number $n$ of unit balls. However, after we apply the symmetry reduction technique, we find an explicit LMI whose size is independent of $n$, and depends only on.the dimension $d$ of the unit ball. In principle, one can use higher levels of the hierarchy to obtain a tighter inner approximation of $P_{n,d}$. Unfortunately,  symmetry reduction gets quite complicated if we go above degree $2$. However, as we show, the first level of the hierarchy unexpectedly provides an asymptotically tight approximation. We note that good behavior of low-degree sum of squares relaxation is a subject of interest in theoretical computer science in relation to the Unique Games Conjecture \cite{o2008optimal}.

We also point out that a small LMI description does not necessarily lead to a simple set of inequalities describing a given set. For instance, some cones of sums of squares are known to have small LMI descriptions, but the semilagebraic inequalities describing them have very large degree \cite{MR2999301}. Therefore, the description of Theorem \ref{thm:maindual2}
is quite fortuitous and not guaranteed from general theory.

The rest of the paper is structured as follows. In the next Section, we present the exact results, i.e., an analogue of the half-degree principle. In Section \ref{sec:asymp}, we provide the asymptotic approximation results. In Section \ref{sec: criteria}, we dualize polynomial results to get asymptotically tight criteria on moments. We conclude in Section \ref{sec:last} with an application of our results to the case $d=3$ and recover the necessary conditions shown in \cite{PhysRevLett.99.250405}, along a set of sufficient conditions.

\section{Exact Results}

In this section, we provide an analogue of the half degree principle for the cone $P_{n,d}$.  More specifically, we show that to test nonnegativity of $Q\in V_{n,d}$  on $K_n^d$ we only need to test that it takes nonnegative values on points $\mathbf{x}=(\mathbf{x}_1, \cdots, \mathbf{x}_n)$ on $K_n^d$ with at most $2d$ of them distinct. Observe that such points form a submanifold of $K_n^d$ of dimension $2d(d-1)$, and this is independent of $n$, while dimension of $K_n^d$ is $dn$. As warm-up and illustration we start with the one-dimensional case $d=1$.

\subsection{One Dimensional Case $d=1$} As a simple first step we characterize the cones $C_{n,1}$ and $P_{n,1}$. The relevant unit ball is $D^1=[-1,1]$. A point in $K_n^1$ is represented by $(x_1, x_2, \cdots, x_n)$ where $x_i \in [-1,1]$. Quadratic polynomial $Q\in V_{n,1}$ has the form:
\begin{equation*}
Q = A_0 + A_1 s_1 + A_{11}s_{11}.
\end{equation*}
Note that $Q$ has only linear terms in each variable $x_i$. Therefore, extreme values of $Q$ occur when $x_i =\pm 1$. In other words, $Q$ is nonnegative on $K_n^1$ if and only if it is nonnegative on the discrete hypercube $H_n=\{-1,+1\}^n$. For a point  $\mathbf{x}=(x_1,\cdots,x_n) \in H_n$ with $k$ entries equal to $-1$ and  $n-k$ entries equal to $+1$ we have $s_1(\mathbf{x})=n-2k$ and $s_{11}(\mathbf{x}) = (n-2k)^2-n$. We immediately obtain the following Proposition:

\begin{prop}\label{prop:d1}
A polynomial $Q=  A_0 + A_1 s_1 + A_{11}s_{11} \in P_{n,1}$ if and only if
$$A_0+A_1(n-2k)+A_{11}((n-2k)^2-n)\geq 0,$$
holds for $k=0,1,\cdots,n$.
\end{prop}

Each inequality represents a side of the polygon shown in \figref{FIG1}. The dual cone $C_{n,1}$ of moment sequences coming from measures is also a polyhedral cone defined by $n+1$ inequalities. The defining inequalities of $C_{n,1}$ follow from Proposition \ref{prop:d1}:

\begin{cor}
A vector $\mathbf{m}=(m_0, m_1, m_{11})\in C_{n,1}$ if and only if
$$m_0(n-1+(n-1-2k)^2)-m_1(n-1-2k)+m_{11} \geq 0,$$
holds for $k=0,1,\cdots,n$.
\end{cor}
Thus, when $d=1$, $C_{n,1}$ and $P_{n,1}$ are both basic semi-algebraic sets, and are given by $n+1$ linear inequalities. 

\subsection{General Dimension}

When $d>1$, $C_{n,d}$ is the conical hull of a semi-algebraic set. Indeed,
\begin{equation*}
C_{n,d} = \operatorname{ConicalHull}\left\{ (1, s_1(\mathbf{x}), \cdots, s_d(\mathbf{x}), s_{11}(\mathbf{x}), \cdots, s_{dd}(\mathbf{x})): \\ \mathbf{x} \in K^d_n \right\}.
\end{equation*}
$K_n^d$ is a basic semi-algebraic set and therefore, its image under a polynomial function is semi-algebraic. A polynomial $Q\in P_{n,d}$ is linear in each of its arguments and therefore, it is nonnegative on $K_n^d$ if and only if it is nonnegative on its boundary, $\partial K_n^d=S^{d-1}\times \cdots \times S^{d-1}$. Therefore, membership of a polynomial $Q$ in $P_{n,d}$ is validated by verifying its non-negativity on an $n(d-1)$ dimensional manifold. However, in theorem \ref{thm:halfdeg}, we show that it suffices to verify its non-negativity on finitely many copies ($O(n^{2d-1})$) of a $2d(d-1)$ dimensional manifold. This theorem is an analogue of the degree principle \cites{MR1996126, MR2864859}. See also \cite{MR3424059} and \cite{MR3471187} for related results. %The following proposition applies for a special case  of and a precursor to Theorem \ref{thm:halfdeg}. 

%\begin{prop}\label{prop:A}
%A polynomial $Q=  A_0 + \sum_{\alpha}A_{\alpha} s_{\alpha} + A\sum_{\alpha}s_{\alpha\alpha} \in P_{n,d}$ if and only if
%$$A_0-nA +\sum_{\alpha}A_{\alpha}x_{\alpha}+A\sum_{\alpha}x_{\alpha}^2\geq 0,$$
%holds for all $\mathbf{x}=(x_1, \cdots, x_d) \in nD^d$, i.e., ball of radius $n$ in $\mathbb{R}^d$.
%\end{prop}
%\begin{proof}
%It follows easily from the observation $s_{\alpha\alpha}= s_{\alpha}^2 - \sum_{i}\xi_{i, \alpha}^2$ and $\sum_{\alpha}\xi_{i, \alpha}^2 =1$ that $Q = A_0-nA + \sum_{\alpha}A_{\alpha}s_{\alpha}+A\sum_{\alpha}s_{\alpha}^2$. Furthermore, the vector $(s_1, \cdots, s_d) = \sum \mathbf{x}_i \in n D^d$ and any vector in $nD^d$ can be written as a sum of $n$ vectors, each in $S^{d-1}$. \end{proof}

We now state and prove the main theorem of this section:

\begin{thm}\label{thm:halfdeg}A polynomial $Q \in V_{n,d}$ is nonnegative on $S^{d-1}\times \cdots \times S^{d-1}$ if and only if $Q(\mathbf{x}_1, \cdots, \mathbf{x}_n)$ is nonnegative for all sets of $n$ points $\mathbf{x}_1, \cdots, \mathbf{x}_n$ on $S^{d-1}$ with at most $2d$ of them distinct.
\end{thm}

\begin{proof} We will prove this theorem via an application of Lagrange multipliers. Recall that $Q\in V_{n,d}$ has the form $$Q = A_0+\sum_{\alpha=1}^d (A_{\alpha}s_\alpha+A_{\alpha\alpha}s_{\alpha\alpha}).$$
Observe that polynomials $Q$ all of whose global minima on $(S^{d-1})^n$ contain more than $2d$ distinct vectors form an open subset of $V_{n,d}$, since a sufficiently small perturbation of a polynomial with all minima containing more than $2d$ distinct vectors results in a polynomial with all minima still containing more than $2d$ distinct vectors. Additionally, polynomials with all coefficients $A_{\alpha}$ and $A_{\alpha\alpha}$ distinct also form an open subset of $V_{n,d}$. Therefore, if a counterexample to the Theorem exists, it can be chosen with $A_{\alpha}$ and $A_{\alpha\alpha}$ distinct, and thus it suffices to prove the Theorem for polynomials $Q$ with all coefficients $A_{\alpha}$ and $A_{\alpha\alpha}$ distinct.

Let $\mathbf{x}^*=(\mathbf{x}_1^*, \cdots, \mathbf{x}_n^*)$ be a global minimum of $Q$ on $\left(S^{d-1}\right)^n$ and let $\mathbf{x}_i^*=(\xi_{i,1}, \cdots, \xi_{i, d})$. Since the global minimum is a critical point, it satisfies the following Lagrange multiplier equations for $\alpha=1,\cdots,d$ and $i=1,\cdots,n$:
\begin{equation*}\label{lagrange_multiplier}
A_{\alpha}+2A_{\alpha \alpha}s_{\alpha}(\mathbf{x}^*) =(\lambda_i+2A_{\alpha \alpha})\xi_{i,{\alpha}},
\end{equation*}
where $\lambda_i\in \R$ are the Lagrange multipliers. Note that the left-hand side of the above equation is independent of the index $i$. We therefore introduce 
$$R_{\alpha}=A_{\alpha}+2A_{\alpha\alpha}s_{\alpha}(\mathbf{x}^*), \hspace{.4 cm} \alpha=1,\cdots, d.$$
%In terms of $R_{\alpha}$, the equations satisfied by the coordinates of the global minimum are
and see that
\begin{align}\label{lagrange_multiplier_R}
R_{\alpha}=(\lambda_i+2A_{\alpha \alpha})\xi_{i,{\alpha}}
%\lambda_i\xi_{i,{\alpha}}=R_{\alpha}-A_{\alpha \alpha}\xi_{i,{\alpha}}
\end{align}
We first settle the simple case when $R_{\alpha}=0$ for $\alpha = 1, \cdots, d$.  It follows from \eqref{lagrange_multiplier_R} that $\xi_{i, \alpha}=0$ unless $(\lambda_i+2A_{\alpha\alpha})=0$. The latter can hold for at most one value of $\alpha$ for any given $i$ since the coefficients $A_{\alpha\alpha}$ are distinct. Thus, $\xi_{i, \alpha}=0$ for all but one value, $\alpha= \alpha_i$  for which we must have $\xi_{i, \alpha_i}= \pm 1$, since $\sum_{\alpha}(\xi_{i,\alpha})^2 = 1$. Therefore, it follows that at most $2d$ of the vectors $\mathbf{x}_i^*$ can be distinct.

Without loss of generality, we may now assume that $R_1  \neq 0$. This also implies that $\xi_{i, 1}\neq 0$ for $i=1, \cdots, n$. Eliminating $\lambda_i$ from Equations \eqref{lagrange_multiplier_R}, we see that
\begin{equation}\label{elimitation}
(R_1+2(A_{\alpha\alpha}-A_{11})\xi_{i,1})\xi_{i,{\alpha}} =R_{\alpha}\xi_{i, 1} 
\end{equation}
Combining this together with the  equations $\sum_{\alpha =1}^d \xi_{i, \alpha}^2=1$, we see that $t=\xi_{i, 1}$ is a solution of following equation:
\begin{equation*}
\sum_{\beta=1}^{d}(R_{\beta}t)^2\prod_{\alpha\neq \beta}^d (R_1+2(A_{\alpha\alpha}-A_{11})t)^2  =  \prod_{\alpha=1}^d (R_1+2(A_{\alpha\alpha}-A_{11})t)^2\sum_{\beta=1}^{d} \xi_{i, \beta}^2 = \prod_{\alpha=1}^d (R_1+2(A_{\alpha\alpha}-A_{11})t)^2.
\end{equation*}
In other words, $\xi_{i, 1}$ is a root of the polynomial
\begin{equation*}\label{Ch_polynomial}
P(t)= \sum_{\beta=1}^{d}(R_{\beta}t)^2\prod_{\alpha\neq \beta}^d (R_1+2(A_{\alpha\alpha}-A_{11})t)^2  -  \prod_{\alpha=1}^d (R_1+2(A_{\alpha\alpha}-A_{11})t)^2
\end{equation*}
for every $i$. Note that the leading term of $P(t)$ is $t^{2d} \times R_1^2\Pi_{\alpha=2}^d4(A_{\alpha\alpha}-A_{11})^2$ and the constant term is $-R_1^{2d}$. Therefore, $P(t)$ has degree $2d$, and thus $\xi_{i,1}$ can take at most $2d$ distinct values. Next we show how a vector $\mathbf{x}_i^*$ can be reconstructed from $\xi_{i,1}$.

Observe from \eqref{elimitation} that $\xi_{i,\alpha}$ is determined by $\xi_{i,1}$ as
\begin{equation*}\label{solution_xi}
\xi_{i, \alpha} = \frac{R_{\alpha}\xi_{i,1}}{R_1+2(A_{\alpha\alpha}-A_{11})\xi_{i,1}}.
\end{equation*}
unless  $R_1+2(A_{\alpha\alpha}-A_{11})\xi_{i,1} =0$ for some $\alpha$. Such an $\alpha=\alpha_0$, if it exists, has to be unique for a given value of $\xi_{i,1}$ since all coefficients $A_{\alpha\alpha}$ are distinct.  Therefore, if there is such an $\alpha_0$ for a given root $\xi_{i, 1}$, the coordinate $\xi_{i, \alpha}$ is uniquely determined by $\xi_{i,1}$ for all $\alpha$ excluding $\alpha_0$. The coordinate $\xi_{i, \alpha_0}$ is determined up to a sign from $\sum_{\alpha =1}^d \xi_{i, \alpha}^2=1$. Therefore, every root of $P(t)$ produces at most two distinct vectors $\mathbf{x}_i$. We next show that those roots that produce two vectors are indeed double roots. 

As argued above, if $\xi_{i, 1}$ is a root of $P(t)$ that produces two vectors then there is one $\alpha_0$ for which $R_1+2(A_{\alpha_0\alpha_0}-A_{11})\xi_{i,1} =0$ and Equation \eqref{elimitation} implies that $R_{\alpha_0}=0$. With this, we can deduce that $(R_1+2(A_{\alpha_0\alpha_0}-A_{11})t)^2$ divides $P(t)$ and therefore, $\xi_{i, \alpha}$ is a double root of $P(t)$. Therefore, a simple root of $P(t)$ produces a unique vector $\mathbf{x_i}$, while a multiple root of $P(t)$ produces at most two vectors $\mathbf{x_i}$. Since degree of $P(t)$ is at most $2d$, it follows that at most $2d$ of vectors $\mathbf{x_i}$ are distinct.

%in all, each root $t^*$ of $P(t)$ produces at most one vector $\mathbf{x}^*_i$ and thus, there can be at-most $2d$ of them distinct. 
\end{proof}

\begin{remark}
There are $\binom{n+2d-1}{n}$ distinct ways of populating a set $\{\mathbf{x}_1, \cdots, \mathbf{x}_n\}$ using $2d$ distinct points on $S^{d-1}$. Therefore, this theorem reduces the search space of non-negativity of $Q$ from an $n(d-1)$ dimensional manifold to $\binom{n+2d-1}{n}$ copies of a $2d(d-1)$ dimensional manifold. 
\end{remark}

\section{Sum of Squares approximation of Nonnegative Polynomials}\label{sec:asymp}

In this section, we develop an asymptotic approximation of $C_{n, d}$ and its dual cone $P_{n, d}$. We begin with the illustrative case of $d=1$, where we develop the main ideas of the proof.

\subsection{The case $d=1$} \label{sec:ad1}
Although the case $d=1$ has a complete solution, i.e., membership of a quadratic $Q=A_0+A_1s_1+A_{11}s_{11}$ in $P_{n,1}$ can be checked by the $n+1$ inequalities in Proposition \ref{prop:d1}, it is a convenient platform to illustrate the ideas that lead us to the main result of this section. As shown in Proposition \ref{prop:d1},  a quadratic $Q$ lies in $P_{n,1}$ if and only if:
\begin{equation*}
A_0-nA_{11} + \sqrt{n}A_1 \left(\frac{n-2k}{\sqrt{n}}\right)  + nA_{11} \left(\frac{n-2k}{\sqrt{n}}\right)^2  \geq 0
\end{equation*} 
for $k=0,1,\cdots, n$. Note that the above expression has been cast in a way that is suggestive of introducing a new variable $X$ to take the place of $\frac{n-2k}{\sqrt{n}}$ and a polynomial $$P_Q(X) = A_0-nA_{11} + \sqrt{n}A_1X+nA_{11}X^2$$ We may rewrite the above conditions as $P_Q(X) \geq 0$ for $X=-\sqrt{n}, -\sqrt{n}+2/\sqrt{n}, \cdots, +\sqrt{n}$. In other words, $Q\in P_{n,1}$ if and only if $P_Q(X)\geq 0$ on $n+1$ evenly spaced points in $[-\sqrt{n}, \sqrt{n}]$. Given that the spacing $2/\sqrt{n}$ approaches zero as $n$ approaches infinity, we are prompted to define the following cone:
\begin{equation*}
\Sigma_{n,1} = \left\{(A_0, A_1, A_{11}): \ P_Q(X) \geq 0 \ \text{for all} \ X\in [-\sqrt{n}, \sqrt{n}]\right\}
\end{equation*}
Clearly we have $\Sigma_{n,1}\subseteq P_{n,1}$. Since $\Sigma_{n,1}$ is a set of univariate quadratics nonnegative on a closed interval, the cone $\Sigma_{n,1}$ is easily characterized and inclusion in it provides a sufficient condition for inclusion in $P_{n,1}$. We next show that a necessary condition can be obtained by slightly enlarging the cone $\Sigma_{n,1}$.
\begin{equation*}
\Sigma_{n,1}' = \left\{ (A_0, A_1, A_{11}): \ \left(\frac{n}{n-1}A_0, A_1, A_{11}\right) \in \Sigma_{n,1} \right\}
\end{equation*} 
Clearly, $\Sigma_{n,1}\subset \Sigma_{n,1}'$. Further, we show, in Proposition \ref{prop:d11}, that $\Sigma_{n,1}'$ contains $P_{n,1}$. In the limit of large $n$, the cones $\Sigma_{n,1}$ and $\Sigma'_{n,1}$ converge.
\begin{proposition}\label{prop:d11} $  \Sigma_{n,1} \subseteq P_{n,1} \subseteq \Sigma_{n,1}'$
\end{proposition}
\begin{proof}
We show that all the extreme rays of $P_{n,1}$ are in $\Sigma_{n,1}'$. Let $Q = (A_0, A_{1}, A_{11}) \in P_{n,1}$ span an extreme ray. It follows that the corresponding polynomial $P_Q(X)$ takes a zero value at two consecutive points in the set $\{-\sqrt{n}, -\sqrt{n}+2/\sqrt{n}, \cdots, +\sqrt{n}\}$. The zeros of $P_Q$ are therefore separated by $2/\sqrt{n}$. Consequently, the minimum value of $P_Q$ is $-A_{11}$. Therefore, $P_Q+ A_{11}$ is nonnegative on $[-\sqrt{n}, \sqrt{n}]$. Also it follows from Proposition \ref{prop:d1} that $A_{11} \leq  \frac{A_0}{n-1}$ for any polynomial $Q\in P_{n, 1}$.  Thus, $P_Q +  \frac{A_0}{n-1}$ is nonnegative on $[-\sqrt{n}, \sqrt{n}]$. Thus, $Q \in \Sigma_{n,1}'$.
\end{proof}
\begin{figure}
\includegraphics[scale=0.45]{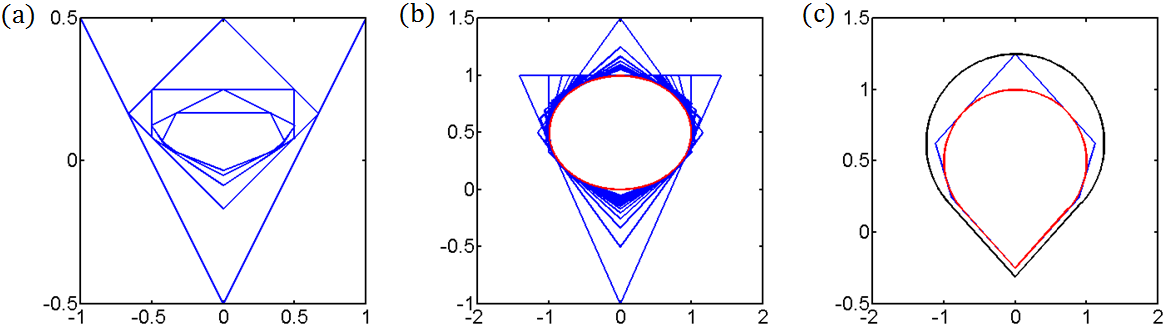}
\caption{(a) shows the cross sections of $P_{n,1}$ for $n= 2$ to $n= 6$. (b) shows the cross sections of the rescaled cones $\tilde{P}_{n,1}$ for $n= 2$ to $n= 20$ in blue and the cross section of the limiting cone $\tilde{\Sigma}$ in red. (c) shows $\tilde{P}_{5,1}$ in blue and the corresponding approximation, $\tilde{\Sigma}_{5,1}$ in red. The expanded cone, $\tilde{\Sigma}_{5,1}'$ is shown in black.}\label{FIG1}
\end{figure}
\noindent The cones $\Sigma_{n,1}$ and $\Sigma_{n,1}'$ that sandwich $P_{n,1}$ are better understood after a rescaling. Taking a cue from the definition of the polynomial $P_Q$, let us define $$\tilde{P}_{n,1} = \left\{\left(B_0, B_1, B_{11}\right): \ \left(B_0, \frac{B_1}{\sqrt{n}}, \frac{B_{11}}{n}\right) \in P_{n,1}\right\}  $$ $$\tilde{\Sigma}_{n,1} = \left\{(B_0, B_1, B_{11}): \ \left(B_0, \frac{B_1}{\sqrt{n}}, \frac{B_{11}}{n}\right) \in \Sigma_{n,1}\right\}  $$ and $$\tilde{\Sigma}_{n,1}' = \left\{(B_0, B_1, B_{11}): \ \left(B_0, \frac{B_1}{\sqrt{n}}, \frac{B_{11}}{n}\right) \in \Sigma_{n,1}'\right\}  $$
This rescaling enlarges them so that we may visualize their limiting behavior. In particular, $\tilde{\Sigma}_{n,1}= \left\{(B_0, B_1, B_{11}): \ B_0-B_{11} + B_1 X + B_{11}X^2 \geq 0, \hspace{.3cm} \text{for all} \hspace{3mm}  X \in [-\sqrt{n}, \sqrt{n}]\right\}$ satisfies $\tilde{\Sigma}_{n+1,1} \subseteq \tilde{\Sigma}_{n, 1}$, i.e., they are nested and the limiting cone $\tilde{\Sigma} = \cap_n \tilde{\Sigma}_{n,1}$ is the cone of all globally nonnegative quadratics, and has a simple characterization: $(B_0,B_1,B_{11}) \in \tilde{\Sigma}$ if and only if \begin{equation}\label{eqn:1cone}\begin{bmatrix} B_0-B_{11}& \frac{1}{2}B_1\\ \frac{1}{2}B_1&B_{11}\end{bmatrix}\succeq 0.
\end{equation}
Figure \ref{FIG1} shows the cross sections of $P_{n,1}$, as well as that of the rescaled cones $\tilde{P}_{n,1}$, $\tilde{\Sigma}_{n,1}$, $\tilde{\Sigma}_{n,1}'$, $\tilde{\Sigma}$. Note that the cones $\tilde{\Sigma}_{n,1}'$ are also nested and the corresponding limiting cone is also $\tilde{\Sigma}$. In this sense, the sufficient condition for inclusion in $P_{n,1}$, provided by $\Sigma_{n,1}$ and the necessary condition provided by $\Sigma_{n,1}'$ approach each other asymptotically. In the next section we generalize the cones $\Sigma_{n,1}$ and $\Sigma_{n,1}'$ to higher dimensions.

\subsection{General Dimension}
We proceed along similar lines as for the $d=1$ case to obtain a necessary and asymptotically sufficient condition for membership in $P_{n, d}$.  We refer to the polynomial $Q$ by its $(2d+1)$-tuple of rescaled coefficients. Analogous to the variable $X$ of the previous section, we define variables $X=(X_1,\cdots,X_d)$ and $Y=(Y_1,\cdots,Y_d)$ as:
\begin{equation*}
\begin{split}
X_{\alpha} &= \sum_{i=1}^n\frac{\xi_{i,\alpha}}{\sqrt{n}}, \hspace{.3cm} \alpha=1,\cdots,d,\\
Y_{\alpha} &= \sqrt{\sum_{i=1}^n\frac{\xi_{i,\alpha}^2}{n} - \frac{X_{\alpha}^2}{n}}, \hspace{.3cm}\alpha=1,\cdots,d,\\
\end{split}
\end{equation*}
with $\mathbf{x}=(\mathbf{x}_1, \cdots, \mathbf{x}_n)$ and $\mathbf{x}_i = (\xi_{i,1}, \cdots, \xi_{i, d})$. Note that the Cauchy-Schwartz inequality ensures that $Y_{\alpha}$ are well defined real numbers. In terms of these variables, we may re-write $Q$ as:
\begin{equation*}
Q(\mathbf{x}_1, \cdots, \mathbf{x}_n) = A_0 + \sum_{\alpha=1}^d \sqrt{n}A_{\alpha}X_{\alpha} + (n-1)\sum_{\alpha=1}^d A_{\alpha\alpha}X_{\alpha}^2 - n\sum_{\alpha=1}^d A_{\alpha\alpha}Y_{\alpha}^2 = P_Q(X, Y).
\end{equation*}
Note that $P_Q(X)$ is defined analogously to the previous section. For $\mathbf{x}\in K_n^d$ the variables $X_{\alpha}$ and $Y_{\alpha}$ satisfy $$\sum_{\alpha=1}^d Y_{\alpha}^2 + \frac{1}{n}X_{\alpha}^2 = ||Y||^2 +\frac{1}{n}||X||^2  \leq 1,$$ which follows from the conditions $||\mathbf{x}_i||\leq 1$.  

We now consider a convex cone in $V_{n,d}$ consisting of polynomials $Q$ such that $P_Q(X,Y)$ is nonnegative for all $(X,Y)$ such that $||Y||^2 +\frac{1}{n}||X||^2  \leq 1$. From the above discussion it follows that this cone lies inside $P_{n,d}$. We show in Theorem \ref{thm:sos} that this cone turns out to coincide with the sum-of-squares relaxation of $P_{n,d}$, defined as $\Sigma_{n,d}$ in the Introduction. To recall this definition, for $i=1,\cdots,n$ let $p_i=1-\sum_{j=1}^d \xi^2_{j,i}$. We defined the cone $\Sigma_{n,d}$ consisting of polynomials $Q(\mathbf{x})$ which can be written as:
$$Q(\mathbf{x})=\sum_{i=1}^r \ell_i^2(\mathbf{x})+\sum_{i=1}^d \lambda_i(1-p_i(\mathbf{x})),$$
where $\ell_i$ are linear polynomials and $\lambda_i \geq 0$. 

\begin{theorem}\label{thm:sos}
A polynomial $Q$ lies in $\Sigma_{n,d}$ if and only if $P_Q(X,Y) \geq 0$ for all  $X, Y$ such that $||Y||^2 + \frac{1}{n}||X||^2 \leq 1$. 
\end{theorem}

\begin{remark}\label{rem:psd}
By the $S$-Lemma \cite{MR2353804} we can express the condition of nonnegativity of $P_{Q}(X,Y)$ as a linear matrix inequality in the coefficients of $P$: $P=A_0 + \sum_{\alpha=1}^d \sqrt{n}A_{\alpha}X_{\alpha} + (n-1)\sum_{\alpha=1}^d A_{\alpha\alpha}X_{\alpha}^2 - n\sum_{\alpha=1}^d A_{\alpha\alpha}Y_{\alpha}^2$ is nonnegative for all  $X, Y$ such that $||Y||^2 + \frac{1}{n}||X||^2 \leq 1$ if and only if there exists $c\geq 0$ such that:
$$\begin{bmatrix} 
A_0 & \frac{\sqrt{n}}{2}A_1&\cdots & \frac{\sqrt{n}}{2}A_d&0&\cdots &0\\
\frac{\sqrt{n}}{2}A_1 & (n-1)A_{11}& 0& 0&0&\cdots&0\\
\vdots     & 0 & \ddots&0& 0&\cdots&0\\
\frac{\sqrt{n}}{2}A_d & 0&0 & (n-1)A_{dd}& 0&\cdots&0\\
0&0&\cdots &0&-nA_{11}&\cdots&0\\
\vdots&0&\cdots&0&0&\ddots&0\\
0&0&\cdots&0&0&0&-nA_{dd}\\ 
\end{bmatrix} +c\begin{bmatrix} 
-1 & 0&\cdots & 0&0&\cdots &0\\
0 & \frac{1}{n}& 0& 0&0&\cdots &0\\
\vdots     & 0 & \ddots&0&0&\cdots &0\\
0 & 0&0 & \frac{1}{n}&0&\cdots &0\\ 
0&0&\cdots &0&1&\cdots&0\\
\vdots&0&\cdots&0&0&\ddots&0\\
0&0&\cdots&0&0&0&1\\ 
\end{bmatrix}
$$
is positive semidefinite.%\gb{Why is the coefficient $-n$ in the top corner of the second matrix?}.
\end{remark}

\begin{proof}
We use symmetry reduction technique of \cite{MR2067190,BR}. Let $V$ be the vector space of linear polynomials in $d\cdot n$ variables $\mathbf{x}_i=(\xi_{i,1},\cdots,\xi_{i,d})$, $i=1,\cdots n$. The symmetric group $S_n$ acts on $V$ by permuting the vectors $\mathbf{x}_i$, or equivalently $S_n$ permutes $n$ groups of $\alpha$-th coordinates, $\xi_{i,\alpha}$. Recall that irreducible $S_n$-modules are indexed by partitions of $n$. For more background on representation theory of $S_n$ see \cite{MR1824028}.

It is not hard to see that $V$ decomposes as follows into irreducible $S_n$-modules:
$$V=(d+1)S^{(n)}\oplus dS^{(n-1,1)}.$$

The $S_n$-invariant part of $V$ corresponding to the partition $(n)$ is spanned by polynomials $1$ and $s_{\alpha}$ for $\alpha=1,\cdots, d$. For the partition $(n-1,1)$ we can split the isotypic component into $d$ irreducible modules by considering for fixed $\alpha=1,\cdots,d$ polynomials of the form $$\sum_{i=1}^n c_i \xi_{i,\alpha} \hspace{.5cm} \text{with} \hspace{.5 cm} \sum c_i=0.$$
We pick $d$ isomorphic representatives $f_{\alpha}=n\xi_{1,\alpha}-s_{\alpha}$, i.e. these polynomials generate the above irreducible modules, and can be mapped to each other under $S_n$-equivariant maps.

Define polynomials $p_{\alpha\beta}=\sum_{i=1}^n \xi_{i, \alpha}\xi_{i, \beta}$ and $q_{\alpha\beta}=\sum_{1\leq i<j\leq n}\xi_{i,\alpha}\xi_{j,\beta}$, as well as $p_{\alpha\alpha}=\sum_{i=1}^n \xi_{i,\alpha}^2$. We will use $\sym$ to denote the \textit{Reynolds operator} :
$$\sym f=\frac{1}{n!}\sum_{\sigma \in S_n} \sigma(f).$$
We observe that  $$\sym \left(\sum_{\alpha=1}^d \lambda_{\alpha}(1-p_{\alpha}(\mathbf{x}))\right)=\frac{\lambda_1+\cdots+\lambda_d}{n}\left(n - \sum_{\alpha=1}^d p_{\alpha\alpha}\right).$$

Now consider $(d+1)\times(d+1)$ matrix $S$ given $S_{\alpha,\beta}=\sym s_{\alpha}s_{\beta}$ with $s_0=1$, and $d\times d$ matrix $F$ given by $F_{\alpha,\beta}=\sym f_{\alpha}f_{\beta}$. The symmetry reduction procedure tells us that $p\in \Sigma_{n,d}$ if and only if there exist positive semidefinite symmetric matrices $G, H$ and $c \geq 0$ such that 
$$p=\langle S,G \rangle+\langle F,H \rangle+c\left(n-\sum_{\alpha=1}^dp_{\alpha\alpha}\right).$$
Further examining matrices $S$ and $F$, we have for $\alpha\neq \beta$  $$ s_{\alpha}s_{\beta}=p_{\alpha\beta}+q_{\alpha\beta}\hspace{.3cm}\text{and} \hspace{.3cm} \sym f_{\alpha} f_{\beta}=\frac{1}{n-1}q_{\alpha\beta}-p_{\alpha\beta},$$
while $$\sym s_{\alpha}^2=s_{\alpha\alpha}+p_{\alpha\alpha}  \hspace{.3cm}\text{and} \hspace{.3cm}\sym f^2_{\alpha}=n^2\sym \xi_{1, \alpha}^2- s_{\alpha}^2=np_{\alpha\alpha}-s_{\alpha}^2=(n-1)p_{\alpha\alpha}-s_{\alpha\alpha}.$$
We observe that in order to obtain $p$ of the given form the coefficients of $p_{\alpha\beta}$ and $q_{\alpha\beta}$ must cancel. However this can only happen if the corresponding entries in $G$ and $H$ matrices are zero. Therefore we see that $H$ is a diagonal matrix and $G$ has non-zero off-diagonal entries only in the first row and column.

Let $$G=\begin{bmatrix} 
g_0 & \frac{1}{2}g_1&\cdots & \frac{1}{2}g_d\\
\frac{1}{2}g_1 & g_{11}& 0& 0\\
\vdots     & 0 & \ddots&0\\
\frac{1}{2}g_d & 0&0 & g_{dd}\\
\end{bmatrix},$$
and $H$ be a diagonal matrix with entries $h_{\alpha\alpha}$. We now observe that $p$ has a potentially non-zero coefficient of $p_{\alpha\alpha}$ for $\alpha=1,\cdots, d$. This coefficient may be canceled only by using $n-\sum_{\alpha=1}^d p_{\alpha\alpha}$. Therefore the coefficient $c$ must satisfy $c=g_{\alpha\alpha}+(n-1)h_{\alpha\alpha}$, $\alpha=1,\cdots d$. Solving for $h_{\alpha\alpha}$ we get $h_{\alpha\alpha}=\frac{c-g_{\alpha\alpha}}{n-1}$, with the additional restriction $c\geq g_{\alpha\alpha}$, since $h_{\alpha\alpha}\geq 0$.

By examining $p=\langle S,G \rangle+\langle F,H \rangle+c(n-\sum_{\alpha=1}^dp_{\alpha\alpha}),$ we see that %\gb{This seems wrong somehow. Why is there no $\sqrt{n}$ factor on $B_i$?}
$$A_0=cn+g_0, \hspace{1cm} A_{\alpha}=g_{\alpha} \hspace{1cm} A_{\alpha\alpha}=\left (\frac{ng_{\alpha\alpha}}{n-1}-\frac{c}{n-1}\right)\hspace{1cm} \text {with} \hspace{1cm} c\geq g_{\alpha\alpha}.$$

Solving for $g$'s we get $g_0=A_0-cn$, $g_{\alpha}=A_{\alpha}$ and $ g_{\alpha\alpha}=\frac{n-1}{n}A_{\alpha\alpha}+\frac{c}{n}$. The condition $c\geq g_{\alpha\alpha}$ can be rewritten as $c\geq A_{\alpha\alpha}$.
Therefore $p$ is a sum of squares if and only there exists $c\geq 0$ and $c\geq A_{\alpha\alpha}$, $\alpha=1,\cdots,d$ such that the following matrix is positive semidefinite:

$$\begin{bmatrix} 
A_0 & \frac{1}{2}A_1&\cdots &\frac{1}{2}A_d\\
\frac{1}{2}A_1 & \frac{n-1}{n}A_{11}& 0& 0\\
\vdots     & 0 & \ddots&0\\
\frac{1}{2}A_d & 0&0 & \frac{n-1}{n}A_{dd}\\
\end{bmatrix} +c\begin{bmatrix} 
-n & 0&\cdots & 0\\
0 & \frac{1}{n}& 0& 0\\
\vdots     & 0 & \ddots&0\\
0 & 0&0 & \frac{1}{n}.\\
\end{bmatrix}$$

%$$\begin{bmatrix} 
%A_0-cn & \frac{1}{2}A_1&\cdots &\frac{1}{2}A_d\\
%\frac{1}{2}A_1 & \frac{n-1}{n}A_{11}+\frac{c}{n}& 0& 0\\
%\vdots     & 0 & \ddots&0\\
%\frac{1}{2}A_d & 0&0 & \frac{n-1}{n}A_{dd}+\frac{c}{n}\\
%\end{bmatrix} $$

%\gb{Let's examine this formulation further using Schur complement: $ \frac{n-1}{n}A_{ii}+\frac{c}{n} >0$, $A_0-cn-v^tC^{-1}v\geq 0$. We have lower bounds on $c$:
%$c>-(n-1)A_{ii}$ and $c\geq 0$}

%\gb{$v^tC^{-1}v=\frac{1}{4}\sum \frac{A_i^2}{\frac{n-1}{n}A_{11}+\frac{c}{n}}=\frac{n}{4}\sum \frac{A_i^2}{(n-1)A_{ii}+c}$}

%$$A_0-cn-\frac{n}{4}\sum \frac{A_i^2}{(n-1)A_{ii}+c}\geq 0$$

%%$$f(c)=c+\frac{1}{4}\sum \frac{A_i^2}{(n-1)A_{ii}+c}\leq A_0/n$$

%$$f'(c)=1-\frac{1}{4}\sum \frac{A_i^2}{((n-1)A_{ii}+c)^2}=0$$

%$$\sum \frac{A_i^2}{((n-1)A_{ii}+c)^2}=4.$$
We incorporate condition $c\geq A_{\alpha\alpha}$ by enlarging the matrices so that the following matrix is positive semidefinite for some $c\geq 0$:

$$\begin{bmatrix} 
A_0 & \frac{1}{2}A_1&\cdots & \frac{1}{2}A_d&0&\cdots &0\\
\frac{1}{2}A_1 & \frac{n-1}{n}A_{11}& 0& 0&0&\cdots&0\\
\vdots     & 0 & \ddots&0& 0&\cdots&0\\
\frac{1}{2}A_d & 0&0 & \frac{n-1}{n}A_{dd}& 0&\cdots&0\\
0&0&\cdots &0&-A_{11}&\cdots&0\\
\vdots&0&\cdots&0&0&\ddots&0\\
0&0&\cdots&0&0&0&-A_{dd}\\ 
\end{bmatrix} +c\begin{bmatrix} 
-n & 0&\cdots & 0&0&\cdots &0\\
0 & \frac{1}{n}& 0& 0&0&\cdots &0\\
\vdots     & 0 & \ddots&0&0&\cdots &0\\
0 & 0&0 & \frac{1}{n}&0&\cdots &0\\ 
0&0&\cdots &0&1&\cdots&0\\
\vdots&0&\cdots&0&0&\ddots&0\\
0&0&\cdots&0&0&0&1\\ 
\end{bmatrix}
$$
We can multiply the above sum on both sides by a diagonal matrix with diagonal $(1,\sqrt{n},\cdots,\sqrt{n})$ to see that $p\in \Sigma_{n,d}$ if and only if there exists $c\geq 0$ such that the following matrix is positive semidefinite:
$$\begin{bmatrix} 
A_0 & \frac{\sqrt{n}}{2}A_1&\cdots & \frac{\sqrt{n}}{2}A_d&0&\cdots &0\\
\frac{\sqrt{n}}{2}A_1 & (n-1)A_{11}& 0& 0&0&\cdots&0\\
\vdots     & 0 & \ddots&0& 0&\cdots&0\\
\frac{\sqrt{n}}{2}A_d & 0&0 & (n-1)A_{dd}& 0&\cdots&0\\
0&0&\cdots &0&-nA_{11}&\cdots&0\\
\vdots&0&\cdots&0&0&\ddots&0\\
0&0&\cdots&0&0&0&-nA_{dd}\\ 
\end{bmatrix} +c\begin{bmatrix} 
-n & 0&\cdots & 0&0&\cdots &0\\
0 & 1& 0& 0&0&\cdots &0\\
\vdots     & 0 & \ddots&0&0&\cdots &0\\
0 & 0&0 & 1&0&\cdots &0\\ 
0&0&\cdots &0&n&\cdots&0\\
\vdots&0&\cdots&0&0&\ddots&0\\
0&0&\cdots&0&0&0&n\\ 
\end{bmatrix}
$$
\noindent The Theorem now follows from the Remark \ref{rem:psd} with $c$ in the Remark replaced by $nc$.
\end{proof}

In the following, we prove that the sufficient condition for membership in $P_{n,d}$ provided by $\Sigma_{n,d}$ is asymptotically necessary, following the same line of arguments as before. We define $\Sigma_{n,d}'$, by expanding $\Sigma_{n,d}$:
\begin{equation*}
\Sigma_{n,d}' = \{(A_0, A_{\alpha}, A_{\alpha\alpha}):\ \left(\frac{n}{n-1}A_0, A_{\alpha}, A_{\alpha\alpha}\right) \in \Sigma_{n,d}\}
\end{equation*} 

\noindent We show in Theorem \ref{thm:approx} that, $P_{n,d} \subseteq \Sigma_{n,d}'$ which provides a sufficient condition for membership in $P_{n,d}$. To prove this theorem, we need a technical lemma, which we prove at the end of this section:

\begin{lemma} \label{lemma:approx}Let $n\in \mathbb{N}$ and $X= (X_1, \cdots , X_d) \in  \R^d$ such that $||X||^2 = X_1^2+\cdots +X_d^2 \leq n $. There exists $\mathbf{x}=(\mathbf{x}_1, \cdots, \mathbf{x}_n) \in K_n^d$ with $\mathbf{x}_i=(\xi_{i,1}, \cdots, \xi_{i, d})$ such that:
\begin{equation*}
\begin{split}
&\sum_{i=1}^n \frac{\xi_{i,\alpha}}{\sqrt{n}} = X_{\alpha} \ \text{for}\  \alpha=1,2,\cdots,d\\
&\sum_{i=1}^n\frac{\xi_{i,\alpha}^2}{n} - \frac{X_{\alpha}^2}{n}  =0 \ \text{for}\  \alpha =1,2,\cdots, d-1\\
&\left|\left(\sum_{i=1}^n \frac{\xi_{i,d}^2}{n} - \frac{X_d^2}{n}\right)- \left(1-\frac{||X||^2}{n}\right)\right| \leq \frac{1}{n}\\
\end{split}
\end{equation*}
\end{lemma}

\noindent We are now ready to state and prove our main result of this Section.
\begin{theorem}\label{thm:approx} $\Sigma_{n,d} \subseteq P_{n,d} \subseteq \Sigma_{n,d}'$
\end{theorem}
\begin{proof}
The first inclusion is immediate. To show the second inclusion, let $Q = (A_0, A_{\alpha}, A_{\alpha\alpha})\in P_{n,d}$. Recall that the polynomial $P_Q(X, Y)$ was defined as
\begin{equation*}
P_Q(X, Y) = A_0 + \sum_{\alpha=1}^d \sqrt{n}A_{\alpha}X_{\alpha} + (n-1)\sum_{\alpha=1}^d A_{\alpha\alpha}X_{\alpha}^2 - n\sum_{\alpha=1}^d A_{\alpha\alpha}Y_{\alpha}^2.
\end{equation*}
In order to show that $P_Q+\frac{n}{n-1} A_0 \geq 0$ whenever $||Y||^2 +\frac{1}{n}||X||^2 \leq 1$, we pick a point $(X, Y)$ that satisfies the latter condition, and approximate it using $\mathbf{x}=(\mathbf{x}_1, \cdots, \mathbf{x}_n)$ with the help of Lemma \ref{lemma:approx}.

If $A_{\alpha\alpha} \leq 0$ for  all $1 \leq \alpha \leq d$, then %$P_Q(X,Y)$ is minimized by setting all $Y_\alpha$ to $0$. 
we choose $\mathbf{x}=(\mathbf{x}_1,\cdots,\mathbf{x}_n)$ with $\mathbf{x}_i = \frac{X}{\sqrt{n}}$ for all $1\leq i \leq n$. The condition $||Y||^2 +\frac{1}{n}||X||^2 \leq 1$ ensures that $\mathbf{x}\in K_n^d$. It follows now that
\begin{equation*}
P_Q(X, Y) \geq A_0 + \sqrt{n}\sum_{\alpha=1}^d A_{\alpha}X_{\alpha} + (n-1)\sum_{\alpha=1}^d A_{\alpha\alpha}X_{\alpha}^2 = Q(\mathbf{x}) \geq 0
\end{equation*}
The last inequality follows since $Q(\mathbf{x}) \in P_{n,d}$. 

If $A_{\alpha\alpha}> 0$ for at least one $\alpha$, we assume without loss of generality that $A_{dd} \geq A_{\alpha\alpha}$ for $\alpha = 1, 2, \cdots d-1$. It also follows that $A_{dd}>0$. Clearly, %\gb{I ma not sure why this is needed}
\begin{equation}\label{thm_approx}
P_Q(X, Y) \geq A_0 + \sqrt{n}\sum_{\alpha=1}^d A_{\alpha}X_{\alpha} + (n-1)\sum_{\alpha=1}^d A_{\alpha\alpha}X_{\alpha}^2 - nA_{dd}\left(1-\frac{||X||^2}{n}\right)
\end{equation}
We now use Lemma \ref{lemma:approx} to pick $\mathbf{x'}=(\mathbf{x'}_1, \cdots, \mathbf{x'}_n)$ such that 
\begin{equation*}
\begin{split}
&X_{\alpha}'= \sum_{i=1}^n \frac{\xi_{i,\alpha}}{\sqrt{n}} = X_{\alpha} \hspace{.5cm} \text{for}\hspace{.5cm} \alpha=1,2,\cdots,d\\
&(Y_{\alpha}')^2= \sum_{i=1}^n\frac{\xi_{i,\alpha}^2}{n} - \frac{X_{\alpha}^2}{n}  =0 \hspace{.5cm} \text{for}\hspace{.5cm}  \alpha =1,2,\cdots, d-1\\
&\left|(Y_d')^2-\left(1-\frac{||X||^2}{n}\right)\right|=\left|\left(\sum_{i=1}^n \frac{\xi_{i,d}^2}{n} - \frac{X_d^2}{n}\right)- \left(1-\frac{||X||^2}{n}\right)\right| \leq \frac{1}{n}.\\
\end{split}
\end{equation*}
The above equations enable us to evaluate $Q(\mathbf{x'})$ and we thus obtain
\begin{equation*}
Q(\mathbf{x'}) = A_0 + \sqrt{n}\sum_{\alpha=1}^d A_{\alpha}X_{\alpha} + (n-1)\sum_{\alpha=1}^d A_{\alpha\alpha}X_{\alpha}^2 - nA_{dd}(Y'_d)^2 \geq 0%=P_Q(X,Y)- nA_{dd}(Y'_d)^2.
\end{equation*}
Finally using Equation \eqref{thm_approx}, we obtain %\gb{something is missing here...}
\begin{equation*}
P_Q(X, Y) + A_{dd} \geq Q(\mathbf{x'})\geq 0 
\end{equation*}
It follows from the $d=1$ case that $A_{dd}\leq \frac{A_0}{n-1} $ and therefore, $P_Q+\frac{A_0}{n-1} \geq 0$ whenever $Q \in P_{n,d}$. 
\end{proof}

\begin{rem}\label{rem:psd2}
It follows from Remark \ref{rem:psd} that $(A_0,A_1,\cdots A_d,A_{11},\cdots A_{dd}) \in \Sigma_{n,d}'$ if and only if there exists $c\geq 0$ such that

$$\begin{bmatrix} 
\frac{n}{n-1}A_0 & \frac{\sqrt{n}}{2}A_1&\cdots & \frac{\sqrt{n}}{2}A_d&0&\cdots &0\\
\frac{\sqrt{n}}{2}A_1 & (n-1)A_{11}& 0& 0&0&\cdots&0\\
\vdots     & 0 & \ddots&0& 0&\cdots&0\\
\frac{\sqrt{n}}{2}A_d & 0&0 & (n-1)A_{dd}& 0&\cdots&0\\
0&0&\cdots &0&-nA_{11}&\cdots&0\\
\vdots&0&\cdots&0&0&\ddots&0\\
0&0&\cdots&0&0&0&-nA_{dd}\\ 
\end{bmatrix} +c\begin{bmatrix} 
-1 & 0&\cdots & 0&0&\cdots &0\\
0 & \frac{1}{n}& 0& 0&0&\cdots &0\\
\vdots     & 0 & \ddots&0&0&\cdots &0\\
0 & 0&0 & \frac{1}{n}&0&\cdots &0\\ 
0&0&\cdots &0&1&\cdots&0\\
\vdots&0&\cdots&0&0&\ddots&0\\
0&0&\cdots&0&0&0&1\\ 
\end{bmatrix}
$$
is positive semidefinite.
\end{rem}
We now prove Lemma \ref{lemma:approx}.% and in the next section, we show that $G_{n,d}^*$ is indeed the SOS relaxation of $P_{n,d}$.
\begin{proof}[Proof of Lemma \ref{lemma:approx}]
We construct a point $\mathbf{x}\in K_n^d$ with the claimed properties. Let
\begin{equation*}
\mathbf{x}_i=
\begin{cases}
\left(\frac{X_{1}}{\sqrt{n}}, \frac{X_{2}}{\sqrt{n}},\cdots \frac{X_{d-1}}{\sqrt{n}}, \sqrt{1-\sum_{\alpha =1}^{d-1}\left(\frac{X_{\alpha}}{\sqrt{n}}\right)^2}\right) \ \text{for} \ i=1, 2,\cdots k\\

\left(\frac{X_{1}}{\sqrt{n}}, \frac{X_{2}}{\sqrt{n}},\cdots \frac{X_{d-1}}{\sqrt{n}}, -\sqrt{1-\sum_{\alpha =1}^{d-1}\left(\frac{X_{\alpha}}{\sqrt{n}}\right)^2}\right) \ \text{for} \ i=k+1, k+2,\cdots n-1\\

\left(\frac{X_{1}}{\sqrt{n}}, \frac{X_{2}}{\sqrt{n}},\cdots \frac{X_{d-1}}{\sqrt{n}}, z\right) \ \text{for} \ i=n\\
\end{cases}
\end{equation*}
We make an appropriate choice of $k $ and $z$ in the following. We first observe that 
\begin{equation*}
\sqrt{1-\sum_{\alpha =1}^{d-1}\left(\frac{X_{\alpha}}{\sqrt{n}}\right)^2} \geq \frac{|X_d|}{\sqrt{n}}
\end{equation*}
Each $\xi_{i,d}$ is equal to $\pm \sqrt{1-\sum_{\alpha =1}^{d-1}\left(\frac{X_{\alpha}}{\sqrt{n}}\right)^2}$ we may choose $k$ of them with the positive sign so that the sum of all the $\xi_{i,d}$'s is closest to $\sqrt{n}X^d$. That is,
\begin{equation*}
\left|\sqrt{n}X_{d}-\sum_{i=1}^{n-1}\xi_{i,d}\right| \leq \sqrt{1-\sum_{\alpha =1}^{d-1}\left(\frac{X_{\alpha}}{\sqrt{n}}\right)^2} \leq 1
\end{equation*} 
Let us now set $z=\left(\sqrt{n}X_{d}-\sum_{i=1}^{n-1}\xi_{i,d}\right)$. It follows from the above inequality that this is a valid choice for a point on $K_n^d$. It also follows that $\frac{\sum_i x_i^{\alpha}}{\sqrt{n}}=X_{\alpha}$ for $\alpha =1, 2, \cdots, d$. It remains to show the last inequality in the lemma. Explicitly, 
\begin{equation*}
\frac{\sum_i \xi_{i,d}^2}{n} - \left(\sum_i\frac{\xi_{i,d}}{n}\right)^2=  \left(1-\frac{1}{n}||X||^2\right)  -\frac{1}{n}\left(1-||\mathbf{x}_n||^2\right). 
\end{equation*}
\end{proof}

\section{Necessary and sufficient criteria for moments}\label{sec: criteria}
In this section, using the cones $\Sigma_{n, d}$ and $\Sigma_{n,d}'$, we develop a necessary condition and a sufficient condition for membership of a vector $(z_0, z_1,\cdots, z_d, z_{11}, z_{22},\cdots, z_{dd})$ in the moment cone $C_{n,d}$. We also show that these two conditions approach each other, i.e., the necessary condition is asymptotically sufficient and the sufficient condition is asymptotically necessary. 

Let us also define $(\Sigma^{'}_{n,d})^*$ as the dual of $\Sigma_{n,d}'$. It follows from Theorem \ref{thm:approx} that 
\begin{equation*}
\Sigma^*_{n,d} \supseteq C_{n,d} \supseteq (\Sigma_{n,d}^{'})^*
\end{equation*}
Thus, membership in $\Sigma_{n,d}^*$ is a necessary condition and membership in $(\Sigma_{n,d}^{'})^*$ is a sufficient condition for membership in $C_{n,d}$. In the following we develop inequality criteria to check for membership in these two cones, expressed as Linear Matrix Inequalities (LMI). 

\subsection{Dualizing Sum of Squares Cones}

Recall from Remark \ref{rem:psd} that the cone $\Sigma_{n,d}$ can be characterized as $(A_0,A_{\alpha},A_{\alpha\alpha})\in \Sigma_{n,d}$ if and only if there exists $c\in \RR$ such the following matrix is positive semidefinite
$$M=\begin{bmatrix} 
A_0-c & \frac{\sqrt{n}}{2}A_1&\cdots & \frac{\sqrt{n}}{2}A_d&0&\cdots &0&0\\
\frac{\sqrt{n}}{2}A_1& (n-1)A_{11}+\frac{c}{n} & 0& 0&0&\cdots&0&0\\
\vdots     & 0 & \ddots&0& 0&\cdots&0&0\\
\frac{\sqrt{n}}{2}A_d & 0&0 & (n-1)A_{dd}+\frac{c}{n} & 0&\cdots&0&0\\
0&0&\cdots &0&-nA_{11}+c&\cdots&0&0\\
\vdots&0&\cdots&0&0&\ddots&0&0\\
0&0&\cdots&0&0&0&-nA_{dd}+c&0\\ 
0&0&\cdots&0&0&0&0&c\\
\end{bmatrix} %+c\begin{bmatrix} 
%-1 & 0&\cdots & 0&0&\cdots &0&0\\
%0 & \frac{1}{n}& 0& 0&0&\cdots &0&0\\
%\vdots     & 0 & \ddots&0&0&\cdots &0&0\\
%0 & 0&0 & \frac{1}{n}&0&\cdots &0&0\\ 
%0&0&\cdots &0&1&\cdots&0&0\\
%\vdots&0&\cdots&0&0&\ddots&0&0\\
%0&0&\cdots&0&0&0&1&0\\ 
%0&0&\cdots&0&0&0&0&1\\ 
%\end{bmatrix}
$$

%\gb{$(n-1)x_{11}-nx_{d+1,d+1}=z_{11}$}
Here the condition $c\geq 0$ has been absorbed into the last row and column. We observe that we can write
$$M=cD+A_0M_0+\sum_{\alpha=1}^dA_{\alpha}M_{\alpha}+\sum_{\alpha=1}^d
A_{\alpha\alpha}M_{\alpha\alpha}, \hspace{.2cm}M\succeq 0,$$
where matrices $D,M_{\alpha}$ and $M_{\alpha\alpha}$ are the coefficient matrices, whose entries depend on $n$. Standard semidefinite programming duality \cite[Chapter 1]{BPT} tells us that the dual cone of $\Sigma_{n,d}$ is given by:
$\Sigma_{n,d}^*$ is the set of all points $(z_0,z_1,\cdots,z_d,z_{11},\cdots,z_{dd})$ such that there exists a positive semidefinite $(2d+2)\times(2d+2)$ matrix $X$ such that
$$\langle D,X\rangle=0, \hspace{.4cm} \langle M_0, X\rangle=z_0,  \hspace{.4cm} \langle M_{\alpha}, X\rangle=z_{\alpha},  \hspace{.4cm} \langle M_{\alpha\alpha}, X\rangle=z_{\alpha\alpha}.$$
Here $\langle A,B\rangle$ is the standard trace inner product $\langle A,B\rangle=\operatorname{trace} AB$. By similar considerations and Remark \ref{rem:psd2} the dual cone of $\Sigma_{n,d}'$ is the set of all points $(z_0,z_1,\cdots,z_d,z_{11},\cdots,z_{dd})$ such that there exists a positive semidefinite $(2d+2)\times(2d+2)$ matrix $X$ satisfying
$$\langle D,X\rangle=0, \hspace{.4cm} \frac{n}{n-1}\langle M_0, X\rangle=z_0,  \hspace{.4cm} \langle M_{\alpha}, X\rangle=z_{\alpha},  \hspace{.4cm} \langle M_{\alpha\alpha}, X\rangle=z_{\alpha\alpha}.$$

\noindent Therefore we obtain the following characterization of $C_{n,d}$:

\begin{theorem}\label{thm:maindual}
A vector $(z_0, z_1,\cdots, z_d, z_{11}, z_{22},\cdots, z_{dd})$ lies in $C_{n,d}$ if there exists a positive semidefinite matrix $X$ and such that:

$$\langle D,X\rangle=0, \hspace{.4cm}\frac{n}{n-1} \langle M_0, X\rangle=z_0,  \hspace{.4cm} \langle M_{\alpha}, X\rangle=z_{\alpha},  \hspace{.4cm} \langle M_{\alpha\alpha}, X\rangle=z_{\alpha\alpha}.$$
Moreover, if there doesn't exist a  positive semidefinite matrix $X$ such that:

$$\langle D,X\rangle=0, \hspace{.4cm} \langle M_0, X\rangle=z_0,  \hspace{.4cm} \langle M_{\alpha}, X\rangle=z_{\alpha},  \hspace{.4cm} \langle M_{\alpha\alpha}, X\rangle=z_{\alpha\alpha},$$
then $(z_0, z_1,\cdots, z_d, z_{11}, z_{22},\cdots, z_{dd}) \notin C_{n,d}$.
\end{theorem}

\noindent Analyzing this formulation further and applying Schur complement \cite{MR2160825}, the above conditions can be expressed as a set of $2^d$ algebraic inequalities giving necessary and sufficient conditions for membership in $C_{n,d}$, which asymptotically converge.

\begin{theorem}\label{thm:maindual2}
A non-zero vector $(z_0, z_1,\cdots, z_d, z_{11}, z_{22},\cdots, z_{dd})$ lies in $C_{n,d}$ if $z_0>0$ and
\begin{equation*}
    \sum_{\alpha=1}^d\max \left\{\frac{z_0z_{\alpha\alpha}}{n}, \frac{z_{\alpha}^2}{n}\right\} \leq  z_0^2\left(\frac{n-1}{n}\right)^2+ \sum_{\alpha=1}^d \frac{(n-1)z_0z_{\alpha\alpha}}{n^2}
\end{equation*}
Moreover, if the inequalities below are violated:
\begin{equation*}
  z_0>0  \hspace{.3cm} \text{and} \hspace{.3cm}  \sum_{\alpha=1}^d\max \left\{\frac{z_0z_{\alpha\alpha}}{n-1}, \frac{z_{\alpha}^2}{n}\right\} \leq  z_0^2+ \sum_{\alpha=1}^d \frac{z_0z_{\alpha\alpha}}{n}
\end{equation*}
Then a non-zero vector $(z_0, z_1,\cdots, z_d, z_{11}, z_{22},\cdots, z_{dd}) \notin C_{n,d}$.
\end{theorem}
\begin{proof}
We may write a $(2d+2)\times (2d+2)$ matrix $X$ that satisfies $\langle M_0, X\rangle=z_0,  \hspace{.1cm} \langle M_{\alpha}, X\rangle=z_{\alpha}$ and $ \langle M_{\alpha\alpha}, X\rangle=z_{\alpha\alpha}$ as
$$X = \begin{bmatrix} 
z_0 & \frac{1}{\sqrt{n}}z_1&\cdots &\frac{1}{\sqrt{n}}z_d & 0 &\cdots &0 & 0 \\
\frac{1}{\sqrt{n}}z_1 & x_{11}& \cdots& x_{1d} & 0 &\cdots &0 & 0\\
\vdots    & \vdots & \ddots&\vdots & \vdots &\vdots &\vdots & \vdots\\
\frac{1}{\sqrt{n}}z_d & x_{1d}&\cdots & x_{dd} & 0 &0 &0 & 0\\
 0 &0 &\cdots & 0 & \frac{n-1}{n}x_{11}- \frac{1}{n}z_{11} &0 &0 & 0 \\
  \vdots &\vdots &\cdots & \vdots & \vdots &\ddots &\vdots & \vdots \\
   0 &0 &\cdots & 0 & 0 &\cdots &\frac{n-1}{n}x_{dd}- \frac{1}{n}z_{dd} & 0 \\
    0 &0 &\cdots & 0 & 0 &\cdots &0 & x_0\\
\end{bmatrix}$$
Note that all the off-diagonal entries that can be chosen to be zero without disturbing the condition $X\succeq 0$ have been chosen to be zero. The final equation, $\langle X, D\rangle =0$ eliminates the variable $x_0$ i.e.,  $x_{0} - z_0 +\sum x_{\alpha\alpha} - \frac{1}{n}\sum z_{\alpha\alpha} =0 $.  The numbers $z_0, z_{\alpha}, z_{\alpha\alpha}$ are moments only if the numbers $x_{ij}$ can be chosen that the matrix $X$ is positive semidefinite. We list down the constraints on $x_{ij}$. Positive semidefiniteness of the lower $(d+1)\times (d+1)$ block yields the following inequalities, the last of which is an upper bound on $\sum x_{\alpha\alpha}$ and the rest are lower bounds on $x_{\alpha\alpha}$:
\begin{equation*}
    \begin{split}
        x_{\alpha\alpha} &\geq \frac{1}{n-1}z_{\alpha\alpha}, \ \text{for}\ \alpha = 1, \cdots, d\\
        \sum x_{\alpha\alpha} &\leq z_0 + \frac{1}{n}\sum z_{\alpha\alpha}.\\
    \end{split}
\end{equation*}
%\gb{Isn't there $1$ more inequality? $$\frac{1}{n(n-1)}\sum z_{\alpha\alpha}\leq z_0.$$ Also,
%$$z_0 + \frac{1}{n}\sum z_{\alpha\alpha}\geq 0.$$}
The above inequalities imply that $$\frac{1}{n(n-1)}\sum z_{\alpha\alpha}\leq z_0, \hspace{.3cm} \text{and} \hspace{.3cm}z_0 + \frac{1}{n}\sum z_{\alpha\alpha}\geq 0.$$ Additionally we know that $x_{\alpha\alpha}\geq 0$ for $1\leq \alpha\leq d$. If $z_0=0$, then it is clear that $z_\alpha=0$ for $\alpha=1,\cdots,d$, and the above inequalities taken together imply that if $z_0=0$, then $z_{\alpha\alpha}=0$ for $\alpha=1,\cdots,d$. Therefore we may assume that $z_0>0$.

Taking Schur complement with respect to $z_0$ (\cite{MR2160825}) of the top $(d+1)\times (d+1)$ block, we see that $X\succeq 0$ if and only if%\gb{Why can't $z_0$ be $0$? (It shouldn't be from moments perspective) It seems that then $z_i=0$ but $z_{ii}$ can be anything...}
: $$z_0\tilde{X}-\frac{1}{n}zz^T \succeq 0.$$ Here, $\tilde{X}$ is the $d\times d$ block containing the variables $x_{ij}$ for $i, j = 1, \cdots, d$ and $z$ is the vector $(z_1, \cdots, z_d)$. We may now choose the off-diagonals of $\tilde{X}$, $x_{ij}$, for $i \neq j$ so that the off-diagonals of $z_0\tilde{X}-\frac{1}{n}zz^T$ are all zero. That leaves us with the inequalities that come from the non-negativity of the diagonal elements, providing $d$ more upper bounds on $x_{\alpha\alpha}$:
\begin{equation*}
    x_{\alpha\alpha} z_0\geq \frac{z_{\alpha}^2}{n} , \ \text{for}\ \alpha = 1, \cdots, d
\end{equation*}
We thus have two lower bounds each on $x_{11}, \cdots , x_{dd}$,  which are the only remaining variables. We also have one upper bound on $\sum x_{\alpha\alpha}$. We can now eliminate the variables $x_{\alpha\alpha}$ by setting the upper bound higher than every lower bound on $\sum x_{\alpha\alpha}$. There are $2^d$ different lower bounds on $\sum x_{\alpha\alpha}$, coming from two lower bounds each on $x_{\alpha\alpha}$. Thus, after eliminating these variables, we are left with the following conditions, that together form a necessary condition for membership in $C_{n, d}$:
\begin{equation*}\label{necessary}
   z_0>0 \hspace{.3cm} \text{and} \hspace{.3cm}  \sum_{\alpha=1}^d\max \left\{\frac{z_0z_{\alpha\alpha}}{n-1}, \frac{z_{\alpha}^2}{n}\right\} \leq  z_0^2+ \sum_{\alpha=1}^d \frac{z_0z_{\alpha\alpha}}{n}
\end{equation*}
For sufficient conditions we replace $z_0$ by $\frac{n-1}{n}z_0$ in the above inequality:
\begin{equation*}\label{sufficient}
   z_0>0 \hspace{.3cm} \text{and} \hspace{.3cm}  \sum_{\alpha=1}^d\max \left\{\frac{z_0z_{\alpha\alpha}}{n}, \frac{z_{\alpha}^2}{n}\right\} \leq  \frac{(n-1)^2}{n^2}z_0^2+ \sum_{\alpha=1}^d \frac{(n-1)z_0z_{\alpha\alpha}}{n^2}.
\end{equation*}
\end{proof}
 
\begin{rem}\label{rem:limit}
The above inequalities can be naturally written in rescaled moment coordinates $\tilde{z}_{\alpha}=\frac{z_{\alpha}}{\sqrt{n}}$ and $\tilde{z}_{\alpha\alpha}=\frac{z_{\alpha\alpha}}{n}$ as:
$$ z_0>0 \hspace{.3cm} \text{and} \hspace{.3cm}  \sum_{\alpha=1}^d\max \left\{ \frac{n}{n-1}z_0 \tilde{z}_{\alpha\alpha}, \tilde{z}_{\alpha}^2\right\} \leq  z_0^2+ \sum_{\alpha=1}^d z_0\tilde{z}_{\alpha\alpha},$$
for necessary conditions and:
$$ z_0>0 \hspace{.3cm} \text{and} \hspace{.3cm}  \sum_{\alpha=1}^d\max \left\{ z_0 \tilde{z}_{\alpha\alpha}, \tilde{z}_{\alpha}^2\right\} \leq  \frac{(n-1)^2}{n^2}z_0^2+ \frac{n-1}{n}\sum_{\alpha=1}^d z_0\tilde{z}_{\alpha\alpha},$$
for sufficient conditions.

We observe that that as $n$ approaches infinity the necessary and sufficient conditions approach each other, and in the (large particle) limit the cone $C_{n,d}$ is given by inequalities: %\gb{unclear if $z_0\geq0$ is right...}:
$$ z_0>0 \hspace{.3cm} \text{and} \hspace{.3cm}  \sum_{\alpha=1}^d\max \left\{ z_0 \tilde{z}_{\alpha\alpha}, \tilde{z}_{\alpha}^2\right\} \leq  z_0^2+ \sum_{\alpha=1}^d z_0\tilde{z}_{\alpha\alpha}.$$
\end{rem}

\begin{example}
Examining the above limit for $d=1$ we see that the limit of the cones $C_{n,1}$ is characterized is given by inequalities:
$$z_0>0,\hspace{.3cm} \tilde{z}_1^2\leq z_0(z_0+\tilde{z}_{11}).$$

Thus the cone $C_{n,1}$ is dual (up to closure) to the characterization of the limit of the cones $P_{n,1}$ given in Equation \eqref{eqn:1cone} of Section \ref{sec:ad1}, since the above inequalities characterize positive semidefiniteness of the matrix
$$\begin{bmatrix}z_0+\tilde{z}_{11} &\tilde{z}_1\\ \tilde{z}_1& z_0 \end{bmatrix},$$
up to closure.
\end{example}
The above results generalize conditions on moments found in \cite{PhysRevLett.99.250405}. In the next section, we show how the above equations reduce to the one found in \cite{PhysRevLett.99.250405} under the appropriate constraints.

\section{A Physical Example}\label{sec:last}
As an example, we consider the case $d=3$ and use the conditions of Theorem  \ref{thm:maindual2} to define a semialgebraic set of moments contained in the set of moments which have a representing measure. Let $z_0, z_1, z_2, z_3, z_{11}, z_{22}, z_{33}$ be a set of numbers. From Theorem \ref{thm:maindual2} the necessary conditions for membership in $C_{n, 3}$ are given by $z_0>0$ and the following eight inequalities:
\begin{equation}\label{criteria}
    \begin{split}
        \frac{z_1^2}{n}+\frac{z_2^2}{n}+\frac{z_3^2}{n} &\leq z_0^2 + \frac{z_{11}z_0}{n}+\frac{z_{22}z_0}{n}+\frac{z_{33}z_0}{n}\\
        \frac{z_1^2}{n}+\frac{z_2^2}{n}+\frac{z_{33}z_0}{n-1} &\leq z_0^2 + \frac{z_{11}z_0}{n}+\frac{z_{22}z_0}{n}+\frac{z_{33}z_0}{n}\\
        \frac{z_1^2}{n}+\frac{z_{22}z_0}{n-1}+\frac{z_3^2}{n} &\leq z_0^2 + \frac{z_{11}z_0}{n}+\frac{z_{22}z_0}{n}+\frac{z_{33}z_0}{n}\\
        \frac{z_1^2}{n}+\frac{z_{22}z_0}{n-1}+\frac{z_{33}z_0}{n-1} &\leq z_0^2 + \frac{z_{11}z_0}{n}+\frac{z_{22}z_0}{n}+\frac{z_{33}z_0}{n}\\
        \frac{z_{11}z_0}{n-1}+\frac{z_2^2}{n}+\frac{z_3^2}{n} &\leq z_0^2 + \frac{z_{11}z_0}{n}+\frac{z_{22}z_0}{n}+\frac{z_{33}z_0}{n}\\
        \frac{z_{11}z_0}{n-1}+\frac{z_2^2}{n}+\frac{z_{33}z_0}{n-1} &\leq z_0^2 + \frac{z_{11}z_0}{n}+\frac{z_{22}z_0}{n}+\frac{z_{33}z_0}{n}\\
        \frac{z_{11}z_0}{n-1}+\frac{z_{22}z_0}{n-1}+\frac{z_3^2}{n} &\leq z_0^2 + \frac{z_{11}z_0}{n}+\frac{z_{22}z_0}{n}+\frac{z_{33}z_0}{n}\\
        \frac{z_{11}z_0}{n-1}+\frac{z_{22}z_0}{n-1}+\frac{z_{33}z_0}{n-1} &\leq z_0^2 + \frac{z_{11}z_0}{n}+\frac{z_{22}z_0}{n}+\frac{z_{33}z_0}{n}\\
    \end{split}
\end{equation}
In the language of spins, as used in \cite{PhysRevLett.107.180502}, the moments can be expressed in terms of the spin expectation values as
\begin{equation*}
    \begin{split}
        z_1 &= 2\langle J_x \rangle\\
        z_2 &= 2\langle J_y \rangle\\
        z_3 &= 2\langle J_z \rangle\\
        z_{11} &= 4\langle J_x^2 \rangle - n\\
        z_{22} &= 4\langle J_y^2 \rangle-n\\
        z_{33} &= 4\langle J_z^2 \rangle-n\\
    \end{split}
\end{equation*}
The first of the conditions \eqref{criteria} reads $2\Delta J_x^2 + 2\Delta J_y^2 + 2\Delta J_z^2 \geq n$ where $\Delta J_x^2 = \langle J_x^2 \rangle - \langle J_x\rangle^2$. The last one reads  $4\langle J_x^2\rangle + 4\langle J_y^2\rangle +4\langle J_z^2\rangle \leq n(n+2)$. These two conditions are trivial, i.e., satisfied by all quantum states, including entangled states. The remaining six inequalities can be viewed as two sets of permutatively related related inequalities:
\begin{equation*}
    \begin{split}
        4(n-1)(\Delta J_y^2 + \Delta J_z^2 )&\geq n(n-2) + 4 \langle J_x^2 \rangle, \\
        4(n-1)(\Delta J_x^2 + \Delta J_z^2 )&\geq n(n-2) + 4 \langle J_y^2 \rangle, \\
        4(n-1)(\Delta J_y^2 + \Delta J_x^2 )&\geq n(n-2) + 4 \langle J_z^2 \rangle, \\
    \end{split}
\end{equation*}
and
\begin{equation*}
    \begin{split}
        4(n-1)\Delta J_x^2 &\geq  4( \langle J_y^2 \rangle + \langle J_z^2 \rangle ) - 2n \\
        4(n-1)\Delta J_y^2 &\geq  4( \langle J_x^2 \rangle + \langle J_z^2 \rangle ) - 2n \\
        4(n-1)\Delta J_z^2 &\geq  4( \langle J_y^2 \rangle + \langle J_x^2 \rangle ) - 2n \\
    \end{split}
\end{equation*}
This is a set of necessary conditions for non-entanglement, which appeared in \cite{PhysRevLett.99.250405}. %Violating any of the above inequalities guarantees entanglement in the system. 
Violation of at least one of the above inequalities is sufficient for entanglement. The necessary conditions for entanglement are obtained by the violation of the sufficient conditions for membership in the moment cone, which are given by replacing $z_0$ by $\frac{n-1}{n}z_0$. The resulting conditions can also be arranged with regards to permutative relations:
\begin{equation*}
    \begin{split}
        4(n-1)(\Delta J_y^2 + \Delta J_z^2 )&\geq n^2-n-1 + 4 (\langle J_x^2 \rangle+\langle J_y \rangle^2+\langle J_z \rangle^2), \\
        4(n-1)(\Delta J_x^2 + \Delta J_z^2 )&\geq n^2-n-1 + 4 (\langle J_y^2 \rangle+\langle J_x \rangle^2+\langle J_z \rangle^2), \\
        4(n-1)(\Delta J_y^2 + \Delta J_x^2 )&\geq n^2-n-1 + 4 (\langle J_z^2 \rangle+\langle J_y \rangle^2+\langle J_x \rangle^2), \\
    \end{split}
\end{equation*}
And
\begin{equation*}
    \begin{split}
        4(n-1)\Delta J_x^2 &\geq  4( \langle J_y^2 \rangle + \langle J_z^2 \rangle +\langle J_x \rangle^2 ) - (n+1), \\
        4(n-1)\Delta J_y^2 &\geq  4( \langle J_x^2 \rangle + \langle J_z^2 \rangle +\langle J_y \rangle^2 ) - (n+1), \\
        4(n-1)\Delta J_z^2 &\geq  4( \langle J_y^2 \rangle + \langle J_x^2 \rangle  +\langle J_z \rangle^2) - (n+1). \\
    \end{split}
\end{equation*}

In conclusion, the relation between entanglement and real algebraic geometry and the method presented in this paper are quite general and can be used to develop entanglement criteria for other settings of many-body systems \cite{1367-2630-19-1-013027}. For instance one can consider other physically relevant symmetries of the measure such as a system with $2n$ atoms with an $S_n \times S_n$ symmetry, also known as mode-split entanglement \cite{PhysRevLett.112.150501}.

\subsection*{Acknowledgements}

Grigoriy Blekherman was partially supported by NSF grant DMS-1352073. Bharath Hebbe Madhusudhana was supported by NSF grant NSF PHYS-1506294. We thank Michael Chapman and Brian Kennedy for providing insights from a physical perspective. We also thank Matthew Boguslawski, Maryrose Barrios and Lin Xin for useful discussions.

\bibliography{References}
\bibliographystyle{unsrt}

\end{document}